\documentclass[times]{anzsauth}
\usepackage{moreverb}
\usepackage{url}
\usepackage{grffile}
\usepackage{lineno}
\usepackage{lipsum}
\usepackage[UKenglish]{isodate}

\usepackage{amsmath,amsfonts,amsthm}
\usepackage{bbold}
\newtheorem{theorem}{Theorem}
\newtheorem{corollary}{Corollary}
\newtheorem{conjecture}{Conjecture}

\def\volumeyear{2021}

\begin{document}
\cleanlookdateon
\runningheads{Anna Karenina and the Two Envelopes}{RICHARD~D~GILL}
\title{Anna Karenina and The Two Envelopes Problem}
\author{R. D. Gill\addressnum{1}}
\affiliation{Mathematical Institute, Leiden University}
\address{
\addressnum{1} \hspace*{-2ex} Mathematical Insitute,
Niels Bohrweg 1, 2333 CA Leiden, The Netherlands. Email: gill@math.leidenuniv.nl\\Date: 17 March 2021. This version may be thought of as  ``version 5.0''. 
}

\begin{abstract} 
The Anna Karenina principle is named after the opening sentence in the eponymous novel: Happy families are all alike; every unhappy family is unhappy in its own way. The Two Envelopes Problem (TEP) is a much-studied paradox in probability theory, mathematical economics, logic, and philosophy. Time and again a new analysis is published in which an author claims finally to explain what actually goes wrong in this paradox. Each author (the present author included) emphasizes what is new in their approach and concludes that earlier approaches did not get to the root of the matter. We observe that though a logical argument is only correct if every step is correct, an apparently logical argument which goes astray can be thought of as going astray at different places. This leads to a comparison between the literature on TEP and a successful movie franchise: it generates a succession of sequels, and even prequels, each with a different director who approaches the same basic premise in a personal way. We survey resolutions in the literature with a view to synthesis, correct common errors, and give a new theorem on order properties of an exchangeable pair of random variables, at the heart of most TEP variants and interpretations. A theorem on asymptotic independence between the amount in your envelope and the question whether it is smaller or larger shows that the pathological situation of improper priors or infinite expectation values has consequences as we merely approach such a situation.
\end{abstract}


\ack{I'm very grateful to Rianne de Heide, who compiled the bibliography file for me as a first step in a new joint research project.}

\maketitle

\section{Dedication} This paper is dedicated to Adrian Baddeley on the occasion of his 65th birthday. Adrian: we wrote some fantastic papers together, and you have been a great friend ever since we met in Bressanone, 1981, when we were ``Young Statisticians''. I hope we will both remain young statisticians in heart and soul, and stay in touch, for many years to come. I suspect that you were the one who first introduced me to the three door problem, for whose definitive solution see \cite{gill2011mhp}, and you certainly later introduced me to Tom Cover's problem, \cite{cover1987pick}, which is covered in this survey too. 

\section{TEP-1}
\subsection{Introduction}
Here is the (currently) standard form of the Two Envelopes Problem (TEP), as given by \cite{falk2008unrelenting}, who cites Wikipedia for the precise formulation. Wikipedia cites \cite{falk2008unrelenting}, so this is kind of frozen now. I will postpone remarks on the (pre-)history of TEP till near the end of the paper. Writing for probabilists and statisticians I shall move fast through (for us) easy developments. However on the way I will discuss logicians', philosophers', and economists' approaches and thereby call into question the very assumptions that for ``us'' probabilists and statisticians are as natural as the air we breathe, hence taken for granted. Though Bayesians and frequentists may also live in different worlds.

You are given two indistinguishable envelopes, each of which contains a positive sum of money. One envelope contains twice as much as the other. You may pick one envelope and keep whatever amount it contains. You pick one envelope at random but before you open it you are offered the possibility to take the other envelope instead. Now consider the following reasoning:
\begin{enumerate}
\item I denote by $A$ the amount in my selected envelope.
\item The probability that $A$ is the smaller amount is $1/2$, and that it is the larger amount is also $1/2$.
\item The other envelope may contain either $2A$ or $A/2$.
\item If $A$ is the smaller amount the other envelope contains $2A$.
\item If $A$ is the larger amount the other envelope contains $A/2$.
\item Thus the other envelope contains $2A$ with probability $1/2$ and $A/2$ with probability $1/2$.
\item So the expected value of the money in the other envelope is $(1/2) 2A + (1/2)(A/2) = 5A/4$.
\item This is greater than $A$, so I gain on average by swapping.
\item After the switch, I can denote that content by $B$
 and reason in exactly the same manner as above.
\item I will conclude that the most rational thing to do is to swap back again.
\item To be rational, I will thus end up swapping envelopes indefinitely.
\item As it seems more rational to open just any envelope than to swap indefinitely, we have a contradiction.
\end{enumerate}

Notice that the problem is \emph{not} to give a correct proof that there is no point in switching. The problem, which many authoritative writers admit still defeats them, is to explain what is wrong with the arguments given above.

For a mathematician it helps to introduce some more notation. I'll refer to the envelopes as Envelope A and Envelope B, and the amounts in them as $A$ and $B$. Let me introduce $X$ to stand for the smaller of the two amounts and $Y$ to stand for the larger. I think of all four as being random variables; but this includes the situation that we think of $X$ and $Y$ as being two fixed though unknown amounts of money $x$ and $y=2x$: a degenerate probability distribution is also a probability distribution, a constant is also a random variable. It includes the model of a frequentist statistician who imagines (or has been reliably informed that) the organizer of this game repeatedly chooses, according to a fixed probability distribution, a new random amount $X$ to be the smaller of the two; then the other amount is determined as $Y=2X$, and finally by the toss of a fair coin (independent of the two amounts) one is put in Envelope A and the other in Envelope B, defining random variables $A$ and $B$. On the other hand, it also includes the model of a true Bayesian statistician which formally is identical to what I just described, but where the probability law of the random variable $X$ is her subjective prior distribution of the unknown, smaller, amount of money in the two envelopes, in one specific realisation of the game. For her, $x$ is a fixed but unknown positive quantity, and the law of the artificial random variable $X$ encapsulates her prior beliefs about $x$. For the frequentist, $x$ is the actually realised value of a physical random variable $X$. Both he and she know that Envelope A is filled by tossing a fair coin and then putting either $x$ or $y=2x$ in it, and since the calculus of subjectivist probability is the same as the calculus of frequentist probability (Kolmogorov rules!), their mathematical models are identical: only their interpretation is different.

So we have four random variables $X$, $Y$, $A$ and $B$ and it is given that $Y=2X>0$ and that $(A,B)=(X,Y)$ or $(Y,X)$. The assumption that the envelopes are indistinguishable and closed and one is picked at random, translates into the assumption that the event $\{A=X\}$ has probability 1/2, whatever the amount $X$; in other words, the random variable $X$ and the event $\{A=X\}$ are independent. And to repeat what I just stated: the notation does not prejudice the question whether probability is taken in its subjectivist or frequentist interpretation -- do we use probability to represent our (lack of) knowledge, or do we use probability to represent chance mechanisms in the real world?

I consider the argument steps 1--12 together with the structural relationships and probabilistic properties of $A$, $B$, $X$ and $Y$ to be the definition of The Two Envelopes Problem (TEP), or more precisely, The Original Two Envelopes Problem (TEP-1). Just as a successful movie may spawn a series of sequels and occasionally even prequels, TEP has done the same. We must therefore be careful to distinguish between the entire franchise TEP and the original TEP. Moreover, the original TEP did not come out of thin air, but had a history. Think of old movies which the public might have forgotten, but the directors of new movies certainly hadn't.

The alert probabilist will notice that something is going wrong in steps 6 and 7. An expectation value is being computed, but how? Is it a conditional expectation or an unconditional expectation? These are two main interpretations of the intention of the author of 1--12: the author meant to compute the unconditional expectation $E(B)$, or the conditional expectation $E(B|A)$. However the author does not reveal his intention so this is pure guesswork on our side. Curiously, probabilists and statisticians such as \cite{falk2008unrelenting} and \cite{christensen1992bayesian} tend to go for the conditional expectation, while philosophers such as \cite{schwitzgebel2008two} think more often that an unconditional expectation was intended. I will describe the philosopher's choice (and many layperson's choice) first.

\subsection{The philosopher's choice}
Let's explore the philosopher's interpretation first. According to that interpretation we are aiming at computation of $E(B)$ by conditioning on the two cases separately: $X=A$ (envelope A contains the smaller amount of money), $X=B$ (envelope B contains the smaller amount). If that is so, then the rule which we want to use is
$$E(B)=P(A=X)E(B\mid A=X)+P(B=X)E(B\mid B=X).$$
The two situations have equal probability $1/2$, as mentioned in step 6, and those probabilities are then substituted, correctly, in step 7. However according to the this interpretation, the two conditional expectations are screwed up. A correct computation of $E(B\mid A=X)$ is the following: conditional on $A=X$, $B$ is identical to $2X$, so we have to compute $E(2X\mid A=X)=2 E(X\mid A=X)$. But we are told that whether or not Envelope A contains the smaller amount $X$ is independent of the amounts $X$ and $2X$, so $E(X\mid A=X)=E(X)$. Similarly we find $E(B\mid B=X)=E(X\mid B=X)=E(X)$.

Thus the expected values of the amount of money in Envelope B are $2E(X)$ and $E(X)$ in the two situations that it contains the larger and the smaller amount. The overall average is $(1/2)2E(X)+(1/2)E(X)=(3/2)E(X)$. Similarly this is the expected amount in Envelope A.

The clearest exponents of the philosophers' diagnosis of the core of the problem are \cite{schwitzgebel2008two} who write: ``\emph{You would expect less in {\rm Envelope A} if you knew that it was the envelope with less than you would if you knew it was the envelope with more}''.  This is perfectly correct, and I think a very intuitive explanation. In fact, we can easily say something stronger: the expected amount in the second envelope given it's the larger of the two is twice the expected amount given it's the smaller!

As many philosophy authors repeat, the resolution of the paradox is that the writer has committed the sin of equivocation: using the same words to describe different things. However this is equivocation of somewhat subtle concepts. Taking the subjective Bayesian interpretation of our model, we are confusing \emph{our beliefs about $b$}, the amount in the second envelope, in the situation where we imagine being informed that it is the larger amount, from what we imagine our beliefs about it would be if we were to imagine being informed that it is the smaller amount. And at the same time we are making an even more serious equivocation, namely of levels: we are confusing expectation values with actual values.

In my opinion the philosopher's interpretation is very far fetched. However it seems to be a very common way in which also ordinary lay persons interpret the context and intent of the writer. There is a very different way to interpret the intention of the writer of steps 6 and 7 which is far more common in the probability literature. Apparently it comes completely naturally to ``us'' probabilists and statisticians, while it is far too sophisticated ever to occur to ordinary folk.

\subsection{The probabilist's choice}
Since the answers are expressed in terms of the amount in Envelope A, it also seems reasonable to suppose that the writer intended to compute $E(B\mid A)$. Contrary to what many writers imagine, this in no way implies that our player is actually looking in his envelope. The point is that he can imagine what his expectation value would be of the contents of Envelope B, for any particular amount $a$ he might \emph{imagine} seeing in his own Envelope A, \emph{if} he were to take a peek. If it would appear favourable to switch \emph{whatever} that imaginary amount might be, then he has no need to peek in his envelope at all: he can decide to switch anyway.

The conditional expectation $E(B\mid A=a)$ can be computed just as the ordinary expectation, by averaging over two situations, but the mathematical rule which is being used is then
$$E(B\mid A)~=~P(A=X\mid A)E(B\mid A=X,A)+P(B=X\mid A)E(B\mid B=X,A).$$
If this was the writer's intention, then in step 7 he correctly substitutes $E(B\mid A=X,A)=E(2X\mid A=X,A)=E(2A\mid A=X,A)=2A$ and similarly $E(B\mid B=X,A)=A$. But he also takes $P(A=X\mid A)=1/2$ and $P(B=X\mid A)=1/2$, that is to say, the writer assumes that the probability that the first envelope is the smaller or the larger doesn't depend on how much is in it. But it obviously could do! For instance if the amount of money is bounded then sometimes one can tell for sure whether Envelope A contains the larger or smaller amount from knowing how much is in it.

In probabilistic terms, under this interpretation, the writer has mistakenly taken independence of the event $\{X=A\}$ from the amount $A$ as the same as the implicitly given assumption that the event $\{A=X\}$ is independent of the random variable $X$.

\subsection{The heart of the matter}
In probability theory we know that (statistical) independence is symmetric. In particular, it is equivalent to say that $A$ is statistically independent of $\{A=X\}$ and to say that $\{A=X\}$ is statistically independent of $A$. The probabilist's interpretation of the mess was that the writer incorrectly assumed $\{A=X\}$ to be independent of $A$. The philosophers Schwitzgebel and Dever's interpretation was that the writer incorrectly assumed $A$ to be independent of $\{A=X\}$.

One point I'm making is that we have no way of knowing what the original writer was meaning to do. One thing is clear: he is doing probability calculations in a sloppy way.  He is computing an expectation by taking the weighted average of the expectations in two different situations. Either he gets the expectations right but the weights wrong, or the weights right but the expectations wrong (or is there a third possibility?). Is he confusing random variables and possible values they can take? Or conditional expectations and unconditional expectations? Conditional probabilities and unconditional probabilities? That simply cannot be decided. TEP-1 has many cores. And these many cores give some reason for the branching family of variant paradoxes which grew from it.

The analysis so far leads me to the interim conclusion that TEP-1 does not deserve to be called a paradox (and certainly not an unresolved paradox, as many writers in philosophy still insist on claiming): it is merely an example of a screwed-up probability calculation where the writer is not even clear what he is trying to calculate. The mathematics being used appears to be elementary probability theory, but whatever the writer is intending to do, he is breaking the standard, elementary rules. Steps 6 and 7 \emph{together} are inconsistent. One cannot say that one of the steps is wrong and the other is right. One can offer as diagnosis, that the inconsistency is caused by the author giving the same names to different things, or the same symbols to different things. We can't deduce what he is confusing with what. He probably is not even aware of the distinctions.  (However ... in the next section I will show that this interim conclusion is hasty. Maybe the writer was smarter than we give him credit for.) 

But first of all I will present a little theorem which ought to be known in the literature, but which however almost nobody seems to realize is true.  

We saw that both philosophers and probabilists both put their finger on essentially the same point: the random variable $A$ \emph{need not} be independent of the event $\{A=X\}$. We can say something a whole lot stronger. The random variable $A$ \emph{cannot} be independent of the event $\{A=X\}$.

Let me make a side remark here, connected to the parenthetical ``however'' above. Suppose that the writer of TEP is a subjective Bayesian. The intended interpretation of the random variables $X$, $Y$, $A$ and $B$ is therefore that their joint probability distribution represents the writer's prior knowledge or uncertainty about the actual amounts involved. Denote the actual smaller and large amount as $x>0$ and $y=2x$, and denote by $a$ and $b$ the actual amounts in the first and second envelopes. These are fixed, unknown amounts of money. The probability distribution of $X$ encapsulates the writer's prior knowledge about $x$.  From this, his prior knowledge about all four amounts is defined by first defining $Y=2X$ and then defining $A$ and $B$ as follows: independently of $X$, with probability one half, $A=X$ and $Y=B$; with the complementary probability one half, $A=Y$ and $B=X$. Since the mathematics I am about to do assumes I am within conventional probability theory, it follows that I started with a  \emph{proper} probability distribution for $X$. Our Bayesian does not have an improper prior. We will return to the possibility of an improper prior in the next section.

Define the random variable $\Delta$ to be the indicator variable of the event $\{A>B\}$, thus also of the event $\{A = Y\}$. Envelope A contains the larger of the two amounts.
\begin{theorem} The random variables $A$ and $\Delta$ cannot be statistically independent.
\end{theorem}

\begin{proof} Suppose to start with that $A$ and $B$ have finite expectation values. Note that $E(A-B\mid A-B>0)>0$. That's the same, since all expectation values are finite, as $E(A\mid A>B)>E(B\mid A>B)=E(A\mid B>A)$. In the last step we used the symmetry of the joint distribution of $A$ and $B$.

Now if the expectation of $A$ depends on whether $A>B$ or $B>A$ then the distribution of $A$ depends on which is true, or in other words, the random variable $A$ is not stochastically independent of the event $A>B$. Equivalently, the event $A>B$ is not independent of the random variable $A$.

For the general case, choose some strictly increasing map from the positive real line to a bounded interval, for instance, arc tangent. Apply this transformation to both $A$ and $B$ and then apply the argument just given to the transformed variables. The ordering of the variables is unaffected by the transformation. So we find that the transformed variable $A$ is not independent of the event $\{A<B\}$, and this implies the non-independence of $A$ of this event. \qed
\end{proof}

Note that we only used the symmetry of the distribution of $A$ and $B$, and the fact that these variables have positive probability to be different. We did not use their positivity. As we will see at the end of the paper, this little theorem lies at the heart not only of the two envelope paradox but also of a whole family of related exchange paradoxes. In every case, the originators of the paradoxes (or the first to ``solve'' them) have ``explained'' the paradox by doing explicit calculations in a particular case. This always leaves later writers with a feeling that the paradox has not really been solved. Indeed, just giving one example does not prove a general theorem. One swallow does not make a summer.

\cite{samet2004one} seem to be the only mathematicians writing in English on TEP who know the general theorem. Olle H\"aggstr\"om knows it, and has published books  \emph{Slumpens sk\"ordar : str\"ovt\aa g i sannolikhetsteorin} and \emph{Streifz\"uge Durch Die Wahrscheinlichkeitstheorie}, \cite{haggstrom1, haggstrom2}, in Swedish and German respectively. \cite{samet2004one}  prove a weaker result in a more general situation: they do not assume symmetry. Their proof is a little more tricky than ours, but still, not much more than a page and basically elementary too. When one adds the assumption of symmetry their result gives ours. \cite{eckhardt2013} Chapter 8  \emph{The Two-Envelopes Problem}  has some nice mathematical results (which I admit that I have not yet digested), which seem to give the same global messages as this paper.

Our proof showed that for any strictly monotone increasing function $g$ such that $E(g(A))$ exists and is finite, $E(g(A)\mid A<B)<E(g(A))<E(g(A)\mid A>B)$. Approximating a not strictly monotone function by strictly increasing functions and going to the limit, we obtain the same inequalities only possibly not strict for all monotone increasing $g$ with $E(g(A))$ exists and finite. This is the same as saying that the laws of $A$ given $A<B$, of $A$ itself, and of $A$ given $A>B$, are strictly stochastically ordered: for all $a$ $P(A>a \mid A<B)\le P(A>a)\le P(A>a \mid A>B)$ , with strict inequality for some $a$. This observation gives us the following general theorem:

\begin{theorem} Suppose $A$ and $B$ are two random variables, unequal with probability 1, and whose joint distribution is symmetric under exchange of the two variables. Then $$P(A < B \mid A) \ne P( B < A \mid A);$$ in other words, for a set of values of $A$ with positive probability, $$P(A < B \mid A=a ) \ne P( B < A \mid A=a ).$$ Also, the laws of $A$ conditional on $A<B$, unconditional, and conditional on $A>B$ are strictly stochastically ordered (from small to large); in other words, $$P(A>a \mid A<B)\le P(A>a)\le P(A>a \mid A>B)\text{ for all } a,$$ with strict inequality for $a$ with positive probability under the law of $A$.
\end{theorem}

Intuitively, $P(A<B \mid A=a)$ ought to be decreasing in $a$. Simple examples show that this is not necessarily true. However it is true in  a certain average sense. For any $a_0$, the result when averaging over $a<a_0$ is never larger than the result when averaging over $a\ge a_0$, where the averaging is with respect to the appropriately normalized law of $A$. To be precise:
$$E(P(A<B \mid A) \mid A<a_0)\ge P(A<B)=1/2\ge E(P(A<B \mid A)|A\ge a_0)$$
for all $a_0$, with both inequalities strict for some $a_0$.

The just mentioned average ordering of the conditional probabilities $P(A<B \mid A=a)$ and the stochastic ordering of the conditional (given the ordering of $A$ and $B$) and unconditional laws of $A$ are exactly equivalent results, and both are forms of the statement that the random variable $A$ and the indicator variable of the event $\{A>B\}$ are strictly \emph{positive orthant dependent}. Recall that $X$ and $Y$ are positive orthant dependent if for all $x$ and $y$, $P(X\ge x, Y\ge y)\ge P(X\ge x)P(Y\ge y)$; I call the dependence strict if there exist $x$ and $y$ such that the inequality is strict. 

\section{TEP-2}
Just like a great movie, the success of TEP led to several sequels and to a prequel, so nowadays when we talk about TEP we have to make clear whether we mean the original movie TEP-1 or the whole franchise.

However before introducing TEP-2 proper, I'll present some intermediate material belonging formally in TEP-1.

\subsection{The totally ignorant Bayesian}
Are steps 6 and 7 of the TEP argument really inconsistent? Suppose the author is actually a Bayesian and the probability distribution she is using for $X$ summarizes her prior knowledge about this amount of money. Suppose she knows absolutely nothing about it, except that it is positive. In that case, if she knows nothing about $X$, she knows nothing about $cX$, for any positive $c$. In particular, if we know nothing about $X$ then knowing $A$ intuitively gives us no clue at all as to whether it is $X$ or $2X$. 

Now, if knowledge (or lack thereof) can be expressed by probability measures, then the probability measure expressing total ignorance about $X$ and that expressing total ignorance about $cX$ must be the same, for any $c>0$.  The only locally bounded measures on the positive half line invariant under multiplication by just two constants $c>0$ and $c'>0$, both different from 1, and such that the ratio of their logarithms is irrational, are those with Lebesgue density proportional to $1/x$. For instance: $c=2$ and $c'=e$. The only bounded measures on the positive half line invariant under multiplication by any positive number are those with density proportional to $1/x$.

Probability theorists will now retort that there is no proper probability distribution with density proportional to $1/x$, end of story! As a probabilist myself, I must agree that this is the end of the road as far as conventional probability theory is concerned. But this does not mean that the matter is thereby closed. For example, creative physicists repeatedly invent mathematical structures that do not exist in contemporary mathematics, they sometimes turn out to be powerful and effective, and if so, they will eventually be incorporated into new mathematics. That a certain formal mathematical framework for some real world domain (reasoning and decision making under uncertainty) does not hold a representative of a conceptual object belonging to that field could just as well be seen as a defect of standard probability theory. In any case, the standard framework of probability theory does contain arbitrarily close approximations to the improper prior. If the author only meant to write that since she knows almost nothing about $X$, it then follows that given $A$, $\Delta$ is pretty certain to be very close to Bernoulli(1/2), we could not fault steps 6 and 7.

Let me make this reasoning firm and also show where it leads to, namely to a whole class of new TEP paradoxes which I'll call TEP-2. This is where TEP moves from probability theory to mathematical economics. But first we stick within (or very close to) probability theory.

Suppose $X$ has the probability distribution with density $c/x$ on the interval $[\epsilon, M]$, zero outside. An easy calculation shows that the proportionality constant is $c=1/\log(M/\epsilon)$. From this we find that the joint distribution of $(A,\Delta)$ has density $c/(2x)$ on 
$$
[\epsilon, M]\times\{0\}~~\cup~~[2\epsilon, 2M]\times\{1\}
$$
and hence the conditional distribution of $\Delta$ given $A$ is Bernoulli(1/2) for $A=a\in[2\epsilon,M]$, while it is degenerate for $a\in [\epsilon,2\epsilon)\cup(M,2M]$. Note that the probability that the distribution of $\Delta$ given $A$ is \emph{not} Bernoulli(1/2) converges to zero as $\epsilon\to0,M\to\infty$.

Similarly, the discrete uniform distribution on $2^k, k=-M,...,N$ has this property as $M,N\to\infty$, and can be seen as an approximation to the improper prior which is uniform on \emph{all} integer powers (positive and negative) of 2.

Let me give an elementary proof characterizing all probability distributions (proper or improper) such that $A$ and $\Delta$ are independent. This seems to me to be more constructive than giving a proof showing that no proper probability distribution exists with this property. However, since I am working with improper as well as proper distributions I have to be a bit careful with probability theory: I move to measure theory, supposing $X$ is ``distributed'' according to a measure on $(0,\infty)$. We understand, I am sure, what I mean by supposing that $\Delta$ is Bernoulli(1/2), independently of $X$, and now I can define $(A,\Delta)$ as function of $(X,\Delta)$ and this generates an image measure on the range of $(A,\Delta)$ which is simply a copy of half of the original improper distribution of $X$ on $(0,\infty)\times\{0\}$ together with half of the original improper distribution of $2X$ on $(0,\infty)\times\{1\}$. We assume that this measure exhibits independence between $A$ and $\Delta$. But that simply means that the improper distributions of $X$ and of $2X$ are identical. Taking logarithms to base 2 the improper distributions on the whole real line of $\log_2 X$ and of $1+\log_2 X$ are identical. The distribution of $\log_2 X$ is invariant under a shift of size +1 and hence under all integer shifts. Such measures are easy to characterize: place an arbitrary measure on the interval $[0,1)$ and glue together all integer shifts of this measure to a measure on the real line. In semi-probabilistic terms, now using $\{.\}$ to denote the fractional part of a real number, $\{\log_2(X)\}$ and $\lfloor \log_2(X)\rfloor$ are independent, with the integer part being uniformly distributed over all integers, and the fractional part having an arbitrary distribution.

It would be nice to show that all probability distributions of $X$ which have $\Delta$ and $A$ approximately independent, are approximately of this form. The crux of the matter is therefore to choose meaningful notions of both instances of ``approximate''. Also, it would be nice to get rid of the special dependence on the number 2. We could just as well have formulated the two envelopes problem using any other factor, at least, large enough to make exchange seem attractive. If a measure on the real line is invariant under all shifts then it has to be uniform. If it is invariant under two relatively irrational shifts then it is uniform.  If it is locally bounded and invariant under all rational shifts it is uniform. 

So far I only succeeded in deriving some partial results, and will stick with the original problem with the special role of 2.

\begin{theorem}Consider a sequence of probability measures of the random variable $X$ such that $A$ and $\Delta$ are asymptotically independent in the sense that the conditional law of $\Delta$ given $A$ converges weakly to Bernoulli(1/2). Then the total variation distance between the laws of $\log_2(X)$ and $1+\log_2(X)$, which is of course equal to the total variation distance between the laws of $X$ and $2X$, converges to zero. 

Conversely, convergence of the total variation distance between the laws of $X$ and $2X$ to zero, implies the asymptotic independence of $A$ and $\Delta$.
\end{theorem}
\begin{corollary}$\sup_k P(\lfloor \log_2 X\rfloor =k)\to 0$.
\end{corollary}
\begin{corollary} The distance between any two (different) quantiles of the law of $X$ converges to infinity.
\end{corollary}
\begin{corollary}For all $\delta>0$, $P(X<\delta E(X))\to 1$.
\end{corollary}
\begin{conjecture} $A$ and $\Delta$ are asymptotically independent if and only if fractional and whole parts of $\log_2 X$ are asymptotically independent, with the whole part asymptotically uniformly distributed over all integers.
\end{conjecture}

\noindent \textbf{Examples}. Suppose $X$ is continuously uniformly distributed on the interval $[1, N]$. The joint probability distribution of $A$ and $\Delta$ can be visualised in the $(a, \delta)$ plane as mass one half, spread uniformly over the line segment $1 < a < N$, $\delta = 0$, and mass one half spread uniformly over  the line segment $2 < a < 2N$, $\delta = 1$. By a careful but elementary computation one can see that for $a \in [2, N]$, the conditional probability that $A < B$ given $A = a$ is exactly equal to 2/3. Outside that interval it is equal to 0 or 1. As $N$ increases the probability of the event $A \in [2, N]$ converges to 3/4. So $A$ and $\Delta$ are \emph{not} asymptotically independent. The variation distance between the laws of $X$ and $2X$ converges to 3/4. Theorem 2 does not apply, though the statement of the first corollary is true, and hence also of the next two.  On the other hand, if we take $\log_2 X$ continuously uniformly distributed on $[0,N]$, then the asymptotic independence does hold and hence the theorem applies, and also its corollaries. If we replace the continuous uniform distributions by the discrete, the same things can be said. All this is consistent with Conjecture 1. 

\medskip

\noindent \textbf{Remark 1}. \emph{Corollary 3} is going to be used to resolve the (still to be introduced) TEP-2 paradox. As the proof will show, Corollary 3 is a corollary of Corollary 2, which follows from Corollary 1, which follows from the theorem (forwards implication).  

\noindent \textbf{Remark 2}. \emph{Conjecture 1} as it stands is ill-posed. Part of the problem is to extend probability theory and then weak convergence theory to include improper prior distributions and allow them to arise as ``weak limits'' in the new, appropriate sense. The first thing to do is to study more examples.

\medskip

\noindent \textbf{Proof of Theorem 3, forward implication}. To say that the conditional law of $\Delta$ given $A$ converges weakly to the constant law Bernoulli(1/2) means precisely that for any $\epsilon>0$ and $\delta$ there exists an $N_0(\epsilon,\delta)$ such that for all $N\ge N_0$, $P(\left|P(\Delta=1\mid A)-\frac 12\right|>\epsilon)\le\delta$. Recall that everything is defined here through the law of $X$ which is supposed to depend on $N$. For all $N$, $\Delta$ is independent of $X$ and Bernoulli(1/2), and $A=X$ if $\Delta=0$, $A=2X$ if $\Delta=1$. Now if $\left|P(\Delta=1\mid A)-\frac 12\right|\le \epsilon$ then $P(\Delta=0\mid A)/P(\Delta=1\mid A)\le (1+2\epsilon)/(1-2\epsilon)=c$, say. Define $Z=\log_2 X$, let $\mathbb 1$ denote an indicator random variable. We have for all $E$,
$$P(Z\in E)=2P(\log_2 A \in E,\Delta=0)\le2\Biggl(\delta+P\bigl(\log_2 A\in E, \Delta=0,\frac{P(\Delta=0\mid A)}{P(\Delta=1\mid A)}\le c\bigr)\Biggr)$$
$$\le 2\delta + 2E\Biggl(P(\Delta=0\mid A)\mathbb 1\{\log_2 A\in E, \frac{P(\Delta=0\mid A)}{P(\Delta=1\mid A)}\le c)\}\Biggr)$$
$$\le 2\delta + 2cE\Biggl(P(\Delta=1\mid A)\mathbb 1\{\log_2 A\in E, \frac{P(\Delta=0\mid A)}{P(\Delta=1\mid A)}\le c)\}\Biggr)$$
$$\le 2\delta + 2cE\Biggl(P(\Delta=1\mid A)\mathbb 1\{\log_2 A\in E\}\Biggr)$$
$$\le 2\delta + 2c P(\log_2 A \in E, \Delta=1)$$
$$=2\delta + 2c P(Z+1 \in E, \Delta=1)$$
$$=2\delta+\frac{1+2\epsilon}{1-2\epsilon}P(Z+1\in E).$$ It follows that
$$P(Z\in E)-P(Z+1\in E) \le 2\delta+4\epsilon/(1-2\epsilon).$$ On the other hand, reversing the roles of the events $\{\Delta=0\}$ and $\{\Delta=1\}$, and starting from the identity $P(Z+1\in E)=2P(\log_2 A\in E,\Delta=1)$, we obtain in exactly the same way
$$P(Z+1\in E)-P(Z\in E) \le 2\delta+4\epsilon/(1-2\epsilon).$$
Since $E$ was arbitrary this proves the claim that the total variation distance between the laws of $Z$ and of $Z+1$ converges to zero.
\qed

\noindent \textbf{Proof of Theorem 3, reverse implication}. This proof is left to the reader. It requires careful choice of two different sets $E$, for instance, $E_+=\{a: P(\Delta=1\mid A=a) > 1/2+\epsilon\}$ for some $\epsilon>0$, and $E_-=\{a: P(\Delta=1\mid A=a) < 1/2-\epsilon\}$ . \qed

\medskip
\noindent \textbf{Proof of Corollary 1}. If $k_0$ maximizes $P(\lfloor Z\rfloor =k)$ then applying the theorem $m$ times we have the asymptotic equality of $P(\lfloor Z\rfloor =k_0),P(\lfloor Z\rfloor +1=k_0),...P(\lfloor Z\rfloor +m=k_0)$. This implies that $\lim\sup P(\lfloor Z\rfloor =k_0)\le1/(m+1)$. Since $m$ was arbitrary,  it follows that $\max_kP(\lfloor Z\rfloor =k)\to 0$.
\qed

\medskip
\noindent \textbf{Proof of Corollary 2}. It is obvious from Corollary 1, that the distance between two fixed (distinct) quantiles of the distribution of $Z$ must diverge as $N\to\infty$.
\qed

\medskip
\noindent \textbf{Proof of Corollary 3}. Let $z_\alpha$ denote the upper $\alpha$-quantile of the law of $Z=\log_2 X$, defined by $P(Z\ge z_\alpha)\ge \alpha$, $P(Z> z_\alpha)< \alpha$. Fix $\epsilon>0$. On the one hand, $$P(X\le 2^{z_\epsilon})>1-\epsilon.$$
On the other hand, $$E(X)=E(2^Z)\ge\frac \epsilon 2 2^{z_{\epsilon/2}}=\frac\epsilon 2 2^{z_{\epsilon/2}-z_\epsilon}2^{z_\epsilon}.$$
Since $z_{\epsilon/2}-z_\epsilon\to\infty$, it follows that for sufficiently large $N$, $\delta E(X)>2^{z_\epsilon}$ and hence $P(X<\delta E(X))>1-\epsilon$.
\qed

\subsection{TEP-2 proper: Great Expectations}
Now for TEP-2 proper, and a shift to some issues much discussed in mathematical economics and decision theory. It was quickly observed that steps 6 and 7 can't both be correct if we restrict attention to $X$ having a proper probability distribution. (As I just explained, I consider that observation to be a cheap way to resolve the TEP-1). However, it also did not take long for many authors to discover probability distributions of $X$ such that $E(B\mid A=a)>a$ for all $a$, or more concisely, $E(B\mid A)>A$. Thus the paradox appears to be resurrected since there \emph{are} situations in which it appears rational to exchange envelopes without knowledge of the content of your envelope. Here is just one such example: let $X$ be 2 to the power of a geometrically distributed random variable with parameter $p=1/3$; to be precise, $P(X=2^n)=2^n/3^{n+1}$, $n=0,1,2...$. When $A=1$, with certainly $A<B$. For any other possible value of $A$ it turns out that $P(A<B\mid A)=3/5$ and $E(B\mid A)=11A/10>A$ except when $A=1$, when $E(B\mid A)=2>A$.

Equally quickly, it was noticed that such examples always had $E(X)=\infty$. This is necessary, since on taking expectation values again, it follows from $E(B\mid A)>A$ that $E(B)>E(A)$ ... or that $E(B)=E(A)=\infty$. But we know a priori (by symmetry) that $E(B)=E(A)$, and indeed $E(B)=E(A)=3E(X)/2$ since the expected amount in both envelopes together is $3E(X)$. Hence all such examples must indeed have $E(X)=\infty$.

Why does this observation resolve the paradox? Well, because if the expectation values of $A$ and $B$ are infinite, you will always be disappointed with what you get, on choosing and opening either envelope. As Keynes famously said, in the long run we are dead. Why are expectation values supposed to be interesting? Because they are supposed to approximate long run averages. But if the infinitely long run average is infinite, any finite average is disappointing. In the mathematical economics literature, as well as our probability distributions expressing our beliefs we have our utilities expressing our value to be assigned to any outcome. Standard economic theory assumes that utilities are bounded. That is supposed to keep paradoxes from the door.

Well, that is the point of view in mathematical economics. I think it is a somewhat cheap way out. In mathematical models it is often perfectly justified to use probability distributions with infinite ranges, and even with infinite expectation values, as convenient, realistic, legitimate mathematical approximations to real life distributions, even though some would insist that all ``real'' distributions actually have bounded support and definitely finite expectation value. The point is, that that point is irrelevant. The fields of mathematical finance, climatology, meteorology, geophysics abound with examples. (An example from theology is Pascal's wager, though some may dispute that that is a ``real world'' example). The important point is the fact that in the real world it is quite possible for averages of a number of independent observations of $X$ to be always far less than the mathematical expectation value of $X$ with overwhelming probability.  Take a distribution of $X$ on the positive real line with infinite expectation and leading to $E(B\mid A)>A$ and truncate it so far to the right that even a million independent observations from $X$ would hardly ever contain one observation exceeding the truncation value. Call the truncated distribution that of $X'$ and use it instead of $X$ to set up TEP-2. You'll find $E(B\mid A)>A$ with huge probability so step 8 suggests you should switch envelopes. But the gain is illusory, since this is a situation where the average of a huge number of copies of $X$ is still far smaller than their expectation value. Expectation value is no guide to decision, even though everything is as finite as you like.

Some philosophers working on the margins of the foundations of the theory of utility do write papers trying to set up a theory of utility which allows unbounded utilities, and use TEP-2 as a test case for such theories. For the reasons just expressed, I think they are barking up a completely wrong tree.

This is where I also return to my intermediate (between TEP-1 and TEP-2) resolution: the author was perhaps a Bayesian using a prior distribution perfectly appropriate to express almost complete lack of knowledge about $X$. Corollary 3 says that as she must admit to having a tiny bit of information, steps 6 and 7 are only approximately correct, not exactly, but now the resolution of the paradox is that in this situation the expectation value of $X$ is so far to the right of where the bulk of its probability distribution lies, that expectation values are no guide to action. It is step 8 which fails. This is a situation where Keynes has the last word. 

Back to TEP-1: since the writer is not working explicitly in a particular formal framework, we do not know what he or she is trying to do. There is not a unique resolution to the paradox of the type ``step so-and-so fails''. There is not a unique explanation of ``what went wrong''. Looking for one is illusory. Unless we take the higher point of view and say: the writer was trying to do probability theory but without knowing its concepts, let alone its rules, and he or she screwed up big time by not making distinctions which in probability theory are crucial to make. TEP-1 is the kind of reason that formal probability theory was invented. Philosophers who work on TEP-1 without knowing modern (elementary) probability are largely wasting their own time; at best they will reinvent the wheel. (One might argue that anybody who works on TEP-1 is wasting their own time).

\section{TEP-3}
Next we start analysing the situation when we do look in Envelope A before deciding whether to switch or stay. If there is a given probability distribution of $X$ this just becomes an exercise in Bayesian probability calculations. Typically there is a threshhold value above which we do not switch. But all kinds of strange things can happen. If a probability distribution of $X$ is not given we come to the randomized solution of \cite{cover1987pick} where we compare $A$ to a random ``probe''  of our own choosing.

Here is the problem, in Cover's words: \emph{Player 1 writes down any two distinct numbers on separate slips of paper. Player 2 randomly chooses one of these slips of paper and looks at the number. Player 2 must decide whether the number in his hand is the larger of the two numbers. He can be right with probability one-half, by just guessing. It seems absurd that he can do better}. 

Spoiler alert. How can he do better? (Cover does not give the answer, but he does know that there is one). Here it is. Player 2 picks a number with a positive probability density with respect to Lebesgue measure on the real line. For any non-empty interval, there is positive probability that it lies in that interval. Hence there is positive probability that it lies between the two numbers written down by Player 1. Now Player 2 uses his random number as surrogate for ``the other number''. He'll give the right answer when his own number is in between Player 1's numbers, but when his number is outside of the range of Player 1's two numbers, he guesses right with probability one half. His overall probability of getting it right is strictly larger than a half.

\section{TEP-0}

This is of course the ``TEP without probability'' of \cite{Smullyan2012satan}. \emph{Let the amount in the envelope chosen by the player be $A$. By swapping, the player may gain $A$ or lose $A/2$. So the potential gain is strictly greater than the potential loss. But let the amounts in the envelopes be $X$ and $2X$. Now by swapping, the player may gain $X$ or lose $X$. So the potential gain is equal to the potential loss}.

The short resolution is simply: the problem is using the same words (potential gain, loss) to describe different things. But different resolutions are possible depending on what one thinks was the intention of the writer. One can try to embed the argument(s) into counterfactual reasoning. Or one also can point out that the key information that Envelope A is chosen at random is not being used in Smullyan's arguments. So this is a problem in logic and this time an example of screwed up logic. Philosophers have lots of ways to clean up this particular mess.

\section{History}

So far I neglected to mention that TEP was a remake of the 1953 \emph{two-neckties} problem, \cite{kraitchik2006mathematical}, of Maurice Kraitchik (1882-1957), a Belgian mathematician and populariser of mathematics born in Minsk. An earlier (1943) edition of Kraitchik's book ``Mathematical recreations'' exists, I do not know if that one already contains the problem. \emph{Two men are each given a necktie by their respective wives as a Christmas present. Over drinks they start arguing over who has the cheaper necktie. They agree to have a wager over it. They will consult their wives and find out which necktie is more expensive. The terms of the bet are that the man with the more expensive necktie has to give it to the other as the prize. The first man reasons as follows: winning and losing are equally likely. If I lose, then I lose the value of my necktie. But if I win, then I win more than the value of my necktie. Therefore, the wager is to my advantage. The second man can consider the wager in exactly the same way; thus, paradoxically, it seems both men have the advantage in the bet. This is obviously not possible (assuming both prefer the more expensive necktie)}.

Kraitchik's main interests were the theory of numbers and recreational mathematics. The two neckties became \emph{two wallets} with \cite{gardner1982aha} and \emph{two envelopes} with \cite{zabell1988loss}, \cite{zabell1988symmetry}, \cite{nalebuff1988puzzles}, \cite{nalebuff1989puzzles} and \cite{gardner1997penrose}. Zabell gave the wide class of problems the name \emph{exchange paradox}. He explains that he heard of the problem from Steve Budrys of the Odesta corporation, and also that he discussed it with lots of other people. Nalebuff tells that he got it from Hal Varian who got it from Sandy Zabell. Zabell (a subjective Bayesian) starts with introducing a third player, Player C, who fills the two envelopes and gives one to Player A and one to Player B. We are not initially told that C does this ``at random''. Hence the other players' prior beliefs about Player C would certainly influence their own decisions. Zabell does go on to focus on the symmetric case that player C is known to be a neutral referee. Nalebuff focussed on a non-symmetric version now called the Ali and Baba problem. Since my focus is on the symmetric case I do not write out the (simple) details here. He neatly retains the paradox that both Ali and Baba, after \emph{imagining} looking in their envelopes, seem to have a good reason to want to switch with the other. A possible ancestry goes back to a problem proposed by Schr\"odinger, quoted in \cite{littlewood1986littlewood}. A highly disguised appearance of the paradox occurred in \cite{blackwell1951translation}. So in the movie paradigm, TEP is actually a \emph{remake} of an almost forgotten classic.

All the symmetric versions of the problem have exactly the same key feature and the same resolution: there is a pair of random variables $A$, $B$ whose distribution is invariant under exchange. They have positive probability to be different; on conditioning that they are different, we may pretend they are certainly different. Hence by our little Theorem 2 at the end of Section 1, the random variable $A$ cannot be independent of the event $\{A<B\}$, or equivalently, the event $\{A<B\}$ cannot be independent of the random variable $A$. Or ... there is an improper prior lurking behind the scenes, expectations are infinite, and exchange is futile.

\section{Conclusions}

Over the years, frequentist probabilists, Bayesian probabilists, logicians, philosophers, and mathematical economists, have all taken a too narrow view of TEP, blind to the existence of other scientific communities. Obviously, the present author is the first to step outside of the narrow confines of their own discipline! Since probability calculus was invented so as to provide a decent language to enable the world to \emph{move on} from problems like TEP, why do so many philosophers still insist on clumsy pre-probability ``solutions'' which are so vague as to be useless? But how come Martin Gardner couldn't solve TEP? And why did so many biggish names deduce that $X$ must have a uniform distribution on $(0,\infty)$, while in fact it's $\log X$ which must be uniform on $(-\infty,\infty)$, to preserve the validity of steps 6 and 7 (if  the special number ``2'' is made arbitrary)? Why did so many authors take a cheap way out to resolve the paradox? It's clear that many people find TEP \emph{irritating}. It is not a \emph{fun} problem like MHP (Monty Hall problem).

I hope this paper shows that there are both subtle and fascinating aspects to TEP and probably even some more interesting maths, if not philosophy, to be done. I did not succeed in showing that limiting independence of $A$ and $\Delta$ implied that $\lfloor \log_2 X \rfloor$ is asymptotically uniform and asymptotically independent of $\{\log_2 X \}$. I could not do this because I don't yet have a way to express formally what I want to prove, since in the limit I am outside of conventional probability theory.

There are certainly some important lessons to people who build probability models in the real world. One should be wary of infinities, but please let's be wary of them for the good reasons, not for non-reasons.

I think it helps a great deal to bear the Anna Karenina principle in mind, when tackling a logical paradox like TEP. Note that the TEP argument is informal. Steps are partly justified, but not fully justified. In order to ``point a finger'' at the mistake, the steps need to be amplified. But why should there only be one way to amplify the steps of the argument so as to fit in to some logical -- but failing -- argument? And why should the failed argument only fail at one step? The writer does not make explicit within which logical framework he is working. We neither know his assumptions nor his intention. Whatever they are, he must be making a mistake, since his conclusion is self-contradictory. But one cannot say that whatever the context and whatever the intention, the mistake is made at the same place. It is hard to be sure that there are no other reasonable contexts and intentions than those which have appeared so far in the literature. As the paradox evolved and migrated to new fields it mutated as well: from its humble origin in recreational mathematics (where it was invented by experts in number theory so as to confuse amateurs) it mutated and migrated to statistics, mathematical economics and to philosophy.

I also found the Anna Karenina principle very useful when arguing with researchers in the foundations of quantum mechanics, who believe that Bell's theorem is false. The theorem in question states that quantum mechanics is incompatible with ``local realism'' -- the world view of Einstein. The theorem, or a formal mathematical version of it, is clearly correct, and it has stood up to more than fifty years of intense scrutiny and much opposition. Again and again, very smart people come up with counterexamples. There is always a mistake in their counter-example, but they will always deny that that is a mistake. Like a persistent student, they will rewrite their manuscript adding new technical detail correcting the mistake they had made before, and hiding a new one buried deeper still in long computations. For some nice open probability problems in this field with distinctly geometric flavour, and not needing any knowledge of quantum mechanics, see \cite{gill2020}.

I find the analogy with the Aliens movie franchise also useful. TEP tells us how important it is to make distinctions. People who write about TEP should be careful to distinguish TEP-1 from the whole franchise. We have this whole franchise precisely because of the Anna Karenina principle. Anna Karenina meets Aliens on the back of a few envelopes. I am looking forward to new papers on TEP, if necessary shredding my own. Arrogance deserves to be punished.

The bibliography to this paper contains a list of all the papers I have studied while writing this one. Many are not cited in the body of this paper, but they have all influenced in one way or another the whole paper. Many of the books are listed with their date of original publication, but with the publisher which presently provides a ``second (or later) edition''. I first started working on this topic through getting involved in Wikipedia discussions, or perhaps one could better say, fights, which somewhat like many court cases (especially in civil law, but also in criminal law) were typically resolved in favour of editors who could recite at length from the Wikipedia rule book while blind drunk if not asleep. Logic or truth are not criteria which a Wikipedia editor is allowed to use. Instead, the key notions to justify inclusion are ``reliable source'', ``notability'', and ``neutral point of view''. Elementary arithmetic is allowed, but elementary logic is disqualified as being ``own research''. Anyway, I'm especially indepted to the Wikipedia editor ``iNic'' who maintains an extensive bibliography on a Wikipedia talk page\\
{ \small\url{https://en.wikipedia.org/wiki/Talk:Two_envelopes_problem/Literature}}.\\
I particularly like the quote he gives, from \cite{syverson2010opening}, ``Indeed if there is anything inherently unbounded about the two-envelope paradox, it is that each search will uncover at least one more reference''.

The Wikipedia page on TEP is still (March 2020) problematic. Please cite my present paper in many future peer-reviewed publications by yourself, in order that it may become an authoritative source for future wikipedia editors.

Almost absent are papers on the quantum two envelope problem. This is surprising in view of the rich literature on quantum versions of MHP (the Monty Hall or three doors problem), in particular \cite{darianoetal2002}. And what led Schr\"odinger to the problem? The interesting paper \cite{ergodos2014enigma} at least mentions the possibility of a quantum TEP. A recent discovery which I have yet to digest is \cite{cheongetal2017}.

\nocite{*}

\newpage
\raggedright
\frenchspacing
\bibliographystyle{anzsj}
\bibliography{tep-bib}
\end{document}